\documentclass[a4paper, 12pt]{article} 

\pdfoutput=1 
\usepackage{amsmath, amsfonts, amscd, amssymb,  amsthm, mathrsfs, ascmac}
\usepackage{latexsym}
\usepackage{longtable, geometry}
\usepackage[english]{babel}
\usepackage[utf8]{inputenc}
\usepackage{ebgaramond}
\usepackage{graphicx}
\usepackage{bbm}
\usepackage{enumitem}
\usepackage{nicefrac}
\usepackage{bm}
\usepackage[svgnames,psnames]{xcolor}
\usepackage[colorlinks,citecolor=DarkGreen,urlcolor=blue,linkcolor=FireBrick,linktocpage,unicode]{hyperref}
 \usepackage{braket}
\usepackage{xparse}
\usepackage{tikz-cd}
\usetikzlibrary{positioning}
\usepackage{authblk}
\usepackage{url, hyperref}

\newtheorem{Thm}{Theorem}[section]
\newtheorem{Def}[Thm]{Definition}

\newtheorem{Lemm}[Thm]{Lemma}
\newtheorem{Prop}[Thm]{Proposition}

\theoremstyle{definition}

\newtheorem{Rem}[Thm]{Remark}

\newcommand{\be}{\begin{equation}}
\newcommand{\ee}{\end{equation}}

\newcommand{\ba}{\begin{align}}
\newcommand{\ea}{\end{align}}

\newcommand{\ben}{\begin{equation*}}
\newcommand{\een}{\end{equation*}}

\def\i<#1>{\langle #1 \rangle}
\def\l<#1>{\left\langle #1 \right\rangle}
\def\b<#1>{\big\langle #1 \big\rangle}

\newcommand{\la}{\langle}
\newcommand{\ra}{\rangle}
\newcommand{\bs}{\boldsymbol}

\newcommand{\Tr}{\mathrm{Tr}}
\newcommand{\BbbR}{\mathbb{R}}
\newcommand{\BbbN}{\mathbb{N}}
\newcommand{\BbbZ}{\mathbb{Z}}

\newcommand{\BbbC}{\mathbb{C}}

\newcommand{\vepsilon}{\varepsilon}
\newcommand{\vphi}{\varphi}
\newcommand{\no}{\nonumber \\}
\newcommand{\Ex}{\mathbb{E}}

\newcommand{\F}{\mathfrak{F}}

\numberwithin{equation}{section}

\makeatletter
  
  \@addtoreset{equation}{section}
\makeatother

\title{
\bfseries 

A finite temperature version of the Nagaoka--Thouless theorem in the $\mathrm{SU}(n)$ Hubbard model

}

\date{}

\author[1]{Tadahiro Miyao}

\affil[1]{Department of Mathematics,  Hokkaido University

Sapporo 060-0810,  Japan\\

e-mail: \texttt{miyao@math.sci.hokudai.ac.jp}
}

\begin{document}

\maketitle

\begin{abstract}
The Aizenman--Lieb theorem for the $\mathrm{SU}(2)$ Hubbard model expands upon the Nagaoka--Thouless theorem for the ground state to encompass finite temperatures. It can be succinctly stated that the magnetization $m(\beta, b)$ of the system in the presence of a field $b$ surpasses the pure paramagnetic value $m_0(\beta, b)=\tanh(\beta b)$.
In this manuscript, we present an extension of the Aizenman--Lieb theorem to the $\mathrm{SU}(n)$ Hubbard model. Our proof relies on a random-loop representation of the partition function, which becomes accessible when expressing the partition function in terms of path integrals.
\end{abstract}

\section{Introduction}

Strongly correlated fermionic many-body systems are of utmost importance in the field of modern condensed matter physics. Among these systems, a simplified model that captures the behavior of electrons in crystals, initially proposed by Kanamori, Gutzwiller, and Hubbard, has undergone extensive theoretical investigations \cite{Gutzwiller1963,Hubbard1963,Kanamori1963}, earning it the name ``Hubbard model." Despite its straightforward mathematical formulation, the Hubbard model has been substantiated through numerical studies to capture a wide range of phenomena, including metal-insulator transitions and magnetic ordering.
However, the rigorous analysis of the Hubbard model poses significant challenges, and the mathematical proofs of various numerical predictions, which encompass a wide range of phenomena, are often still lacking.
\medskip

 Nevertheless, in the realm of magnetic properties of the ground states of the Hubbard model, mathematical studies have made notable progress, resulting in several established theorems \cite{Lieb2004,Tasaki2020}. The pioneering work of Nagaoka and Thouless focused on a many-electron system characterized by an extremely strong Coulomb interaction and precisely one hole on the lattice. They proved that the ground state of such a system exhibits ferromagnetic behavior \cite{Nagaoka1965,Thouless_1965}. The Nagaoka--Thouless theorem represents the first rigorous result in the study of ferromagnetic ground states in the Hubbard model and has had a profound impact on subsequent research concerning magnetism in this model.
 \medskip
 
Aizenman and Lieb made a significant advancement by extending the Nagaoka--Thouless theorem to finite temperatures using a technique called random loop representations applied to the partition function \cite{Aizenman1990}.
Moreover, the author of the present paper goes even further by extending the scope of the Nagaoka--Thouless and Aizenman--Lieb theorems to systems that involve electron-phonon interactions and interactions with quantized radiation fields \cite{Miyao2017,Miyao2020-2}.  These extensions allow for the investigation of more complex interactions and provide valuable insights into the behavior of correlated electron systems coupled to other degrees of freedom.
 \medskip

 Recent advancements in experimental techniques for ultracold atoms have opened up new possibilities for studying the $\mathrm{SU}(n)$ Hubbard model, which is an extension of the conventional spin-$1/2$ Hubbard model with general $\mathrm{SU}(n)$ symmetry \cite{Cazalilla_2014,PhysRevX.6.021030,Pagano2014,Scazza2014,Zhang2014}. Theoretical investigations of the $\mathrm{SU}(n)$ model have primarily relied on numerical calculations, revealing the emergence of phenomena not observed in systems described by the usual $\mathrm{SU}(2)$ Hubbard model \cite{PhysRevB.77.144520, PhysRevLett.98.160405, Titvinidze2011,Zhao2007}. This extended model provides a rich framework for exploring novel quantum many-body phenomena and offers insights into the behavior of correlated systems with larger spin degrees of freedom.
On the contrary, just like the $\mathrm{SU}(2)$ model, the rigorous analysis of the $\mathrm{SU}(n)$ model poses significant challenges. Consequently, the question of how the rigorous findings of the $\mathrm{SU}(2)$ model can be extended to the $\mathrm{SU}(n)$ model is intriguing from both a physical and mathematical perspective. Regarding investigations in this realm, we refer to the following references: \cite{Katsura2013, LIU20191490, PhysRevB.96.075149, Tamura2021, Yoshida2021}.
In particular, Katsura and Tanaka, in their work \cite{Katsura2013}, meticulously examine the connectivity condition, which serves as the linchpin in the proof of the Nagaoka--Thouless theorem, in order to extend the Nagaoka and Thouless result to the $\mathrm{SU}(n)$ Hubbard model.
\medskip

This paper aims to extend the Aizenman--Lieb theorem, a finite temperature adaptation of the Nagaoka--Thouless theorem, to the $\mathrm{SU}(n)$ Hubbard model. The proof strategy entails the appropriate extension of the random-loop representation, employed in characterizing the $\mathrm{SU}(2)$ model, to analyze the partition function of the $\mathrm{SU}(n)$ Hubbard model.
Random-loop representations for the partition functions of quantum spin systems are widely recognized as potent analytical tools in the rigorous investigation of critical phenomena \cite{cmp/1104270709}. It is anticipated that the constructed random-loop representation in this paper will find further applications.
 \medskip

 The structure of this paper is as follows. In Section \ref{Sec2}, we introduce the $\mathrm{SU}(n)$ Hubbard model and explicitly present the main  theorems of this study. Furthermore, we emphasize that these fundamental theorems constitute extensions of the Aizenman--Lieb theorem. Subsequently, in Section \ref{Sec3}, we establish the necessary groundwork for demonstrating the main theorems by developing the Feynman--Kac--It{\^o} formulas for the heat semigroup generated by the $\mathrm{SU}(n)$ Hubbard model. Moving on to Section \ref{Sec4}, we initially construct a random loop representation for the partition function. Subsequently, utilizing this representation, we furnish proofs of the main theorems stated in Section \ref{Sec2}.

\subsection*{Acknowledgements}
T. M. was supported by JSPS KAKENHI Grant Numbers 20KK0304 and 23H01086. T. M. expresses deep gratitude for the generous hospitality provided by Stefan Teufel at the Department
of Mathematics, University of T\"{u}bingen, where T. M. authored a segment of this manuscript
during his stay.

\subsection*{Declarations}
\begin{itemize}
\item  Conflict of interest: The Authors have no conflicts of interest to declare
that are relevant to the content of this article.
\item  Data availability: Data sharing is not applicable to this article as no
datasets were generated or analysed during the current study
\end{itemize}

\section{Main results}\label{Sec2}
Let $d$ be a natural number greater than or equal to $2$.
Consider a $d$-dimensional hypercube lattice denoted as $\Lambda$, with each side having a length of $2\ell$: $\Lambda=(\mathbb{Z}\cap[-\ell, \ell))^d$.
The ${\rm SU}(n)$ Hubbard Hamiltonian on $\Lambda$ is defined by
\be
H_{\Lambda}^{\rm H}=\sum_{\sigma=1}^n\sum_{x, y\in \Lambda} t_{x, y} c_{x, \sigma}^*c_{y, \sigma}+\sum_{x\in \Lambda} \mu_x n_x
+\sum_{x, y\in \Lambda} U_{x, y} n_xn_y.
\ee
Here, $c_{x, \sigma}^*$ and $c_{x, \sigma}$ represent the creation and annihilation operators of a fermion located at site $x$ with the flavor $\sigma$, respectively, satisfying the customary anti-commutation relations:
\begin{align}
\{c_{x, \sigma}, c_{y, \tau}\}=0=\{c_{x, \sigma}^*, c_{y, \tau}^*\},\ \ \{c_{x, \sigma}, c_{y, \tau}^*\}=\delta_{x, y} \delta_{\sigma, \tau},
\end{align}
where $\delta_{x, y}$ represents the Kronecker delta.
$n_{x}$ is the number operator of fermions at the site $x$:
\be
n_{x}=\sum_{\sigma=1}^n c_{x, \sigma}^*c_{x, \sigma}.
\ee
$T=\{t_{x, y}: x, y\in \Lambda\}$ is a hopping matrix of fermions.  In this paper, we assume that $T$ describes the nearest neighbor hopping:\footnote{Let $x, y\in \Lambda$. We say that $x$ and $y$ are nearest neighbor if 
$\|x-y\|_{\infty}=1$, where $\|x\|_{\infty}=\max_{j=1}^d |x_j|$.} 
\be
t_{x, y}
=\begin{cases}
t & \mbox{if $x$ and $y$ are nearest neighbor}\\
0  & \mbox{otherwise.}
\end{cases}
\ee
We suppose that the Hamiltonian acts on the $N:=|\Lambda|-1$ particle space:
\be
\bigwedge^N \ell^2(\Lambda) \otimes \BbbC^n,
\ee
where $\bigwedge^N$ represents the $N$-fold anti-symmetric tensor product.
For the sake of convenience, we assume that the on-site Coulomb interaction is uniform:
\be
U=U_{x, x}\ \ \mbox{for all $x\in \Lambda$}.
\ee

We introduce the operators that will play a fundamental role in this paper by
\be
h_{\sigma}=N_{\sigma}-N_{\sigma+1},\ \ \ \sigma=1, \dots, n-1,
\ee
where $N_{\sigma}$ is the number operator of fermions with the flavour $\sigma$:
\be
N_{\sigma}=\sum_{x\in \Lambda} n_{x, \sigma},\ \ n_{x, \sigma}=c_{x, \sigma}^*c_{x, \sigma}.
\ee
If we define
\be
e_{\sigma, \sigma\rq{}}=\sum_{x\in \Lambda} c_{x, \sigma}^* c_{x, \sigma\rq{}},\ \ \sigma, \sigma\rq{}=1,\dots, n,\ \ \sigma\neq \sigma\rq{},
\ee
 the family $\{h_{\sigma}, e_{\sigma, \sigma\rq{}}\}$ of operators gives a representation of the $\mathfrak{su}(n)$ Lie algebra.
 Furthermore, the family of operators $\{h_{\sigma}\}$ serves as a representation of the Cartan subalgebra of $\mathfrak{su}(n)$.  Operators representing the Cartan subalgebra are widely acknowledged for their fundamental role in the representation theory of Lie algebras. For a more comprehensive understanding, please refer to \cite{Humphreys1972}.

In the case of $n=2$, $H_{\Lambda}$ is the conventional Hubbard Hamiltonian, and $h_1$ corresponds to the third component $S_{\rm tot}^{(3)}$ of the total spin operators.
Hence, in an extension of the terminology used for the $n=2$ case, we refer to the following Hamiltonian as a model for the interaction between an external magnetic field ${\bs b}=(b_1, \dots, b_{n-1})\in \BbbR^{n-1}$ and fermions:
\be
H^{\rm H}_{\Lambda}({\bs b})=H_{\Lambda}^{\rm H}-\sum_{\sigma=1}^{n-1} b_{\sigma} h_{\sigma}.
\ee

In this manuscript, we explore the system under the condition $U=\infty$. To describe the effective Hamiltonian for such a system, certain preparations are necessary.
Initially, we define an orthogonal projection denoted as $Q_{\Lambda}$ as follows: let $E_{n_x}(\cdot)$ be the spectral measure associated with  $n_x$, and assign $Q_{\Lambda, x}$ as $E_{n_x}(\{1\})$. We can then define
\be
Q_{\Lambda}=\prod_{x\in \Lambda}Q_{\Lambda, x}. \label{DefQ}
\ee
Subsequently, we proceed to define the subspace $\F_N$ within the $N$-particle space as follows:
\be
\F_{N}=Q_{\Lambda}\bigwedge^N \ell^2(\Lambda) \otimes \BbbC^n. \label{DefF_N}
\ee
The Hilbert space $\F_N$ effectively represents the set of state vectors characterizing a system wherein exactly one hole is present in $\Lambda$, while each site, except for the hole site, is occupied by a single fermion.
Under the condition of $U=\infty$, it necessitates an infinite amount of energy for two or more fermions to simultaneously occupy a single site. Consequently, the Hilbert space of states describing such a system is constrained to $\F_N$. Correspondingly, the effective Hamiltonian can be expressed as follows:
\be
H_{\Lambda}({\bs b})=Q_{\Lambda} H_{\Lambda}^{{\rm H}, U=0}({\bs b})Q_{\Lambda},
\ee
where $H_{\Lambda}^{{\rm H}, U=0}({\bs b})$ denotes the linear operator obtained by setting $U=0$ in the definition of $H_{\Lambda}^{{\rm H}}({\bs b})$.

The subsequent proposition provides a mathematical formulation of the aforementioned intuitive explanation:
\begin{Prop}\label{EffH}
In the limit of  $U\to \infty$, $H_{\Lambda}^{\rm H}({\bs b})$ converges to $H_{\Lambda}({\bs b})$ in the following sense:
\be
\lim_{U\to \infty} \Big\| (H_{\Lambda}^{\rm H}({\bs b})-z)^{-1}-(H_{\Lambda}({\bs b})-z)^{-1} Q_{\Lambda} \Big\|=0,\ \ z\in \BbbC\setminus \BbbR,
\ee
where $\|\cdot\|$ stands for the operator norm.
\end{Prop}
Proposition \ref{EffH} can be proved by similar arguments as in  the proof of Theorem 2.5 in \cite{Miyao2017}.

To present the first main theorem, we introduce the partition function for $H_{\Lambda}({\bs b})$ defined as follows:
\be
Z_{\Lambda}(\beta;  {\bs b})=\Tr_{\F_N}\Big[
e^{-\beta H_{\Lambda}({\bs b})}
\Big], \ \ \ \beta\ge 0.
\ee

\begin{Thm}\label{Main1}
Let $P_N$ denote the entire partitions of $N.$\footnote{Thus, each ${\bs n}=\{n_j\}_{j=1}^k\in P_N$ satisfies $n_j\in \BbbN$ and $\sum_{j=1}^k n_j=N$.}
For any ${\bs n}\in P_N$, there exists a positive number $D_{\beta}({\bs n})$ such that 
\be
Z_{\Lambda}(\beta;  {\bs b})=\sum_{{\bs n}\in P_N}D_{\beta}({\bs n}) \mathcal{G}_{\beta}({\bs n}; {\bs b}),\label{ExZ}
\ee
where, for each ${\bs n} =\{n_j\}_{j=1}^k\in P_N$, $\mathcal{G}_{\beta}({\bs n}; {\bs b})$ is defined by
\begin{align}
\mathcal{G}_{\beta}({\bs n}; {\bs b}) &= \prod_{j=1}^k G_{\beta}(n_j; {\bs b}),\label{DefG1}\\
G_{\beta}(m; {\bs b}) &= e^{\beta m b_1}+e^{-\beta m b_{n-1}}+\sum_{\sigma=2}^{n-1} e^{\beta m (-b_{\sigma-1}+b_{\sigma})}.
\end{align}

\end{Thm}

\begin{Rem}
\rm 
For the case of $n=2$, Theorem \ref{Main1} can be expressed as follows: 
\be
Z_{\Lambda}(\beta; {\bs b}) =\sum_{{\bs n} \in P_N} D_{\beta}({\bs n}) \prod_{j=1}^k 2 \cosh(\beta b_1 n_j),\label{AL}
\ee
which reproduces the result established by Aizenman--Lieb \cite{Aizenman1990} for the standard Hubbard model.

\end{Rem}

To present the second result, we introduce the following symbols:
\be
B_{\sigma}=\begin{cases}
b_1 & \mbox{if $\sigma=1$}\\
b_{\sigma}-b_{\sigma-1} & \mbox{if $2\le \sigma \le n-2$}\\
-b_{n-1} & \mbox{if $\sigma=n-1$}.
\end{cases}
\ee

\begin{Thm}\label{Main2}
Let $n\ge 2$.
Fix $\sigma\in \{1, 2, \dots, n-1\}$,  arbitrarily.
Assume that $B_{\sigma}>B_{\tau}$  for all $\tau\neq\sigma$.
Then one obtains 
\be
\la h_{\sigma}\ra_{\beta}\ge  \frac{f_{\beta, \sigma}(\bs b)}{1+g_{\beta, \sigma}({\bs b})} N, \label{MainInq}
\ee
where $\la h_{\sigma}\ra_{\beta}$ stands for the thermal expectation\footnote{To be precise,
$\la h_{\sigma}\ra_{\beta}$ is defined by 
$\la h_{\sigma}\ra_{\beta}=\Tr_{\F_{N}}[h_{\sigma} e^{-\beta H_{\Lambda}({\bs b})}] \Big/ Z_{\Lambda}(\beta; {\bs b})$.
}
 of $h_{\sigma}$,  and
the functions $f_{\beta, \sigma}({\bs b})$ and $g_{\beta, \sigma} ({\bs b})$ are respectively given by  
\begin{align}
f_{\beta, \sigma}({\bs b}) &= \frac{1-e^{-\beta (B_{\sigma}-B_{\sigma+1})}}{1+e^{-\beta (B_{\sigma}-B_{\sigma+1})}},\label{Deff}\\
g_{\beta, \sigma} ({\bs b}) &= \sum_{\tau \neq \sigma, \sigma+1} e^{\beta (B_{\tau}-B_{\sigma})}. \label{Defg}
\end{align}
For $n=2$, we understand that $g_{\beta, \sigma}({\bs b})=0$.

\end{Thm}

\begin{Rem}
\begin{itemize}
\item
By using  \eqref{MainInq} and the fact  $g_{\beta, \sigma}({\bs b}) \le n-2$, we get
\be
\la h_{\sigma}\ra_{\beta} \ge \frac{f_{\beta, \sigma}(\bs b)}{n-1} N. \label{CoroMain}
\ee
Note that this inequality holds even for $n=2$.
\item
For the case of $n=2$, when we assign $B_1=-B_2=b>0$, we can re-establish the result of Aizenman--Lieb \cite{Aizenman1990} from \eqref{CoroMain} as follows:
\be
\la h_1\ra_{\beta}\ge  \tanh(\beta b)N.
\ee

\end{itemize}
\end{Rem}

\section{Functional integral representations for the semigroup generated by the Hamiltonian}\label{Sec3}
\subsection{Preliminaries}
\subsubsection{Case of a single fermion}\label{SigleF}
To prove the main theorems, we require a functional integral representation for $e^{-\beta H_{\Lambda}({\bs b})}$. In this section, we will provide an outline of how to construct the functional integral representation for the $\mathrm{SU}(N)$ Hubbard model, for the convenience of the readers.

States of a single fermion are represented by normalized vectors in the Hilbert space:
\be
\ell^2(\Lambda) \otimes \BbbC^n=\ell^2(\Omega),
\ee
where $\Omega=\Lambda\times \{1, \dots, n\}$. The inner product in $\ell^2(\Omega)$ is  given by 
\be
\la f|g\ra=\sum_{\sigma=1}^n \sum_{x\in \Lambda} f(x, \sigma)^* g(x, \sigma),\ \ f, g\in \ell^2(\Omega).
\ee
We define the free Hamiltonian $h_0$ for  a single fermion as follows:
\be
(h_0f)(x, \sigma)=\sum_{\sigma=1}^n \sum_{y\in \Lambda}t_{x, y}\Big(
f(x, \sigma)-f(y, \sigma)
\Big),\ \ f\in \ell^2(\Omega).
\ee

To represent  fundamental physical observables, we introduce the function $\delta_{(x, \sigma)}\ ((x, \sigma)\in \Omega)$ on $\Omega$ defined as follows: 
\be
\delta_{(x, \sigma)}(y, \tau)=\delta_{x, y} \delta_{\sigma, \tau},\ \ (y, \tau)\in \Omega.
\ee
Now, define the function $k_{\sigma}$  by
\be
k_{\sigma}=\sum_{x\in \Lambda} \{ \delta_{(x, \sigma)}-\delta_{(x, \sigma+1)} \},\ \ \sigma=1, \dots, n-1.
\label{Defk}
\ee
In this paper, we adopt the convention that the multiplication operator by a given function $v$ on $\Omega$ is denoted by the same symbol. Thus, for $v$ and $f$ in $\ell^2(\Omega)$, we have:
\be
(vf)(x, \sigma)=v(x, \sigma) f(x, \sigma).
\ee
Under this convention, 
for 
${\bs b}=(b_1, \dots, b_{n-1}) \in \BbbR^{n-1}$ and a function $\mu : (x, \sigma) \ni\Omega\to \mu_x\in \BbbR$, we define the self-adjoint operator $h_{\mu}({\bs b})$ on $\ell^2(\Omega)$ by
\be
h_{\mu}({\bs b})=h_0+\mu-\sum_{\sigma=1}^{n-1} b_{\sigma} k_{\sigma}.
\ee
Here, $\mu$ represents the on-site potential, and $\bs b$ represents the external field. It should be noted that $\mu_x$ does not depend on $\sigma$.

For the sake of simplicity in notation, we define $\BbbZ_+=\BbbZ\cap [0, \infty)$, which represents the set of non-negative integers.
Let $(Y_n)_{n\in \BbbZ+}$ be a discrete-time Markov chain with the state space $\Omega$. The Markov chain is characterized by the  transition probability:
 \begin{align}
 P(Y_n=X|Y_{n-1}=Y)=
  \delta_{\sigma, \tau} \frac{t_{x, y}}{d(y)}, \ \ \ d(x)=\sum_{y\in \Lambda} t_{x, y}
 \end{align}
 for $n\in \BbbN, X=(x, \sigma)\in \Omega$ and $Y=(y, \tau)\in \Omega$.
In the rest of this paper, we work with  the fixed probability space $(M, \mathcal{F}, P)$.
Let $(T_n)_{n\in \BbbZ_+}$ be independent exponentially  distributed  random variables of parameter $1$, independent of 
$(Y_n)_{n\in \BbbZ_+}$. 
We define the random variables $S_n$ and $J_n$ as follows: 
\begin{align}
S_n=\frac{T_n}{d(Y_{n-1})},\ \ \  J_n=S_1+\cdots +S_n.
\end{align}
Then, we define the process $X_t$ as:
\begin{align}
X_t=\sum_{n\in \BbbZ_+} \mathbbm{ 1}_{\{ J_n\le t < J_{n+1}\}} Y_n,
\end{align}
where $\mathbbm{1}_S$ represents the indicator function of the set $S$.
Under this setting,   $(X_t)_{t\ge 0}$ is  a right continuous process\footnote{To be presice, the topology on $\Omega$  is determined by the norm  $
\|(x, \sigma)\|=\max_{j=1, \dots, n}\{|x_j-y_j|\}+|\sigma-\tau|
$.}. Furthemore, 
$J_0:=0, J_1, J_2, \dots$ are the jump times of $(X_t)_{t\ge 0}$,  and 
$S_1, S_2, \dots$ are the holding times of $(X_t)_{t\ge 0}$.
Let $P_X(\cdot)=P(\cdot | X_0=X)$,  and let $(\mathcal{F}_t)_{t\ge 0}$ be the filtration defined by 
$\mathcal{F}_t=\sigma(X_s| s\le t)$. It can be shown that the process
$
(M, \mathcal{F}, (\mathcal{F}_{t})_{t\ge 0}, (P_X)_{X\in \Omega})
$
 is a strong Markov process, see,  e.g., \cite[Theorems 2.8.1 and 6.5.4]{Norris1997}.

 The following Feynman--Kac--It{\^o}  formula 
is well-known: 
\begin{align}
\big(
e^{-t h_{\mu}({\bs b})} f
\big)(X)=\Ex_X\Big[
e^{-\int_0^tv(X_t) dt}f(X_t)
\Big], \ \ f\in \ell^2(\Omega),\ \ X\in \Omega, \label{FKI1}
\end{align}
where
$\Ex_X[f]$ is the expected value of $f$ with respect to $P_X$, and 
 the function $v$ on $\Omega$  is defined by 
\be
v(X)=\mu(X)-\sum_{\sigma=1}^{n-1}b_{\sigma} k_{\sigma}(X).
\ee
A detailed proof of this formula can be found in \cite{Gneysu2015}. This representation \eqref{FKI1} serves as the first step in our analysis.

\subsubsection{Case of $N$ fermions}
In an $N$-fermion system, the non-interacting Hamiltonian is given by
\begin{align}
T=\underbrace{h_{\mu}({\bs b}) \otimes 1 \otimes \cdots \otimes 1}_N +1 \otimes h_{\mu}({\bs b})\otimes 1 \otimes \cdots \otimes 1 +\cdots+ 1 \otimes \cdots \otimes 1\otimes h_{\mu}({\bs b}), 
\label{NELHAMI}
\end{align}
where $h_{\mu}({\bs b})$ is the self-adjoint operator defined earlier. The operator $T$ acts on the Hilbert space $\bigotimes^N\ell^2(\Omega)$, which is the $N$-fold tensor product of $\ell^2(\Omega)$. In the subsequent analysis, we will freely use the identification:
\be
\bigotimes^N\ell^2(\Omega)= \ell^2(\Omega^N), \label{IdnTens}
\ee
 where $\Omega^N$ represents the $N$-fold Cartesian product of  $\Omega$. 
Note that this identification is implemented by the unitary operator $\iota: \bigotimes^N\ell^2(\Omega) \to \ell^2(\Omega^N)$ given by
\be
\iota( f_1\otimes \cdots \otimes f_N) = \big(f_1(X_1) \cdots f_N(X_N)\big)_{(X_1, \dots, X_N)\in \Omega^N}.
\ee

In the $N$-fermion system, we define the operators $\delta_{(x, \sigma)}^{(j)}$ on $\bigotimes^N_{j=1}\ell^2(\Omega)$ as multiplication operators given by
\begin{align}
\delta_{(x, \sigma)}^{(j)}=\underbrace{1\otimes \cdots \otimes 1 \otimes \overbrace{\delta_{(x, \sigma)} }^{j^{\mathrm{th}}}\otimes 1\otimes  \cdots \otimes 1}_N,\ \ j=1, \dots, N. \label{DefVod}
\end{align}
The Coulomb interaction term is 
 described as  the multiplication operator $V$ defined by 
\begin{align}
V&=V_{\mathrm{o}}+V_{\mathrm{d}},
\end{align}
where
\begin{align}
V_{\mathrm{d}}=U\sum_{x\in \Lambda}\sum_{\sigma, \tau=1}^n\sum_{i, j=1}^N \delta_{(x, \sigma)}^{(i)} \delta_{(x, \tau)}^{(j)},\ \
V_{\mathrm{o}}=\sum_{\sigma, \tau=1}^n\sum_{x, y : x\neq y}\sum_{i, j=1}^N U_{x, y} \delta_{(x, \sigma)}^{(i)}\delta_{(y, \tau)}^{(j)}. \label{VasMulti}
\end{align}
Here, $V_{\mathrm{d}}$ represents the diagonal part of the Coulomb interaction, where each fermion at position $x$ interacts with fermions at the same position $x$. On the other hand, $V_{\mathrm{o}}$ represents the off-diagonal part of the Coulomb interaction, where fermions at different positions $x$ and $y$ interact with each other. The parameter $U$ represents the strength of the on-site interaction, while $U_{x, y}$ represents the interaction potential between two fermions at positions $x$ and $y$.

Based on the previous setups, we define the Hamiltonian $L({\bs b})$ describing the interacting $N$-fermion system as the sum of the non-interacting Hamiltonian $T$ and the Coulomb interaction term $V$:
\be
L({\bs b}) =T+V.
\ee

To incorporate the Fermi--Dirac statistics for the fermions, we introduce the antisymmetrizer operator $A_N$ on $\ell^2(\Omega^N)$. 
The antisymmetrizer $A_N$ is defined as follows:
\begin{align}
(A_NF)({\bs X})=\sum_{\tau\in \mathfrak{S}_N} \frac{\mathrm{sgn}(\tau) }{N!} F(\tau^{-1}{\bs X})    \label{DefAN}
\end{align}
for $F\in \ell^2(\Omega^N)$ and   ${\bs X}=(X^{(1)}, \dots, X^{(N)}) \in \Omega^N$. Here, $\mathfrak{S}_N$ represents the permutation group on the set $\{1, \dots, N\}$, $\mathrm{sgn}(\tau)$ denotes the sign of the permutation $\tau$, and $\tau {\bs X}:=(X^{(\tau(1))}, \dots, X^{(\tau(N))})$ denotes the permuted configuration of ${\bs X}$ under the permutation $\tau$. The operator $A_N$ acts as an orthogonal projection from $\ell^2(\Omega^N)$ onto $\ell^2_{\mathrm{as}}(\Omega^N)$, which is the space of all antisymmetric functions on $\Omega^N$.
Using the antisymmetrizer operator, the Hamiltonian of interest, denoted as $H_{\Lambda}^{\rm H}({\bs b})$, can be expressed as
\be
H_{\Lambda}^{\rm H}({\bs b})=A_N L({\bs b})A_N.
\label{RepHami}
\ee

To construct a Feynman--Kac--It\^o formula for the semigroup generated by $H_{\Lambda}^{\rm H}({\bs b})$, we introduce the set $\Omega_{\neq}^N$, defined as
\begin{align}
\Omega_{\neq}^N=\Big\{{\bs X}\in \Omega^N\, \Big|\, X^{(i)} \neq X^{(j)} \ \mbox{for all $i, j\in \{1, \dots, N\}$ with $i\neq j$}\Big\}.
\end{align}
Using Eq.\eqref{IdnTens}, we have the following identification:
\be
\bigwedge^N \ell^2(\Lambda)\otimes \BbbC^n=\ell_{\rm as}^2(\Omega_{\neq}^N).
\ee
This identification is useful in the construction of the Feynman--Kac--It\^o formula.

In the context of an $N$-fermion system, given ${\bs m}=(m_1, \dots, m_N)\in (M)^N$, we define the random variable $X_s^{(j)}({\bs m})=X_s(m_j)$ for each $j=1, \dots, N$. Here, $M$ represents the sample space of the probability space introduced in Section \ref{SigleF}. This  random variable represents the spatial position and flavor of the $j$-th fermion at time $s$.
A right-continuous $\Omega^N$-valued function $({\bs X}_t({\bs m}))_{t\ge 0}=\big(X^{(1)}_t({\bs m}), \dots, X_t^{(N)}({\bs m})\big)$ is referred to as a {\it path} associated with ${\bs m}$. If we write $X_t^{(j)}({\bs m})=(x_t^{(j)}({\bs m}), \sigma_t^{(j)}({\bs m}))$, then $\sigma_t^{(j)}({\bs m})$ is called the {\it flavor component}  of $X_t^{(j)}({\bs m})$, and $x_t^{(j)}({\bs m})$ is called the {\it spatial component} of $X_t^{(j)}({\bs m})$, respectively. The collection of spatial components $(x_t^{(1)}({\bs m}), \dots, x_t^{(N)}({\bs m}))$ represents a trajectory of the $N$ fermions.

Define  the event  by
\begin{align}
D&=D_{{\rm O}}\cap D_{\mathrm{S}},
\end{align}
where
\begin{align}
D_{{\rm O}}&=
\Big\{{\bs m} \in (M)^N\, \Big|\, 
\mbox{
${\bs X}_s({\bs m}) \in \Omega_{\neq}^N$ for all $s\in [0, \infty)$
}
\Big\} ,\\
D_{{\rm S }}&=\Big\{{\bs m} \in (M)^N\, \Big|\, 
\mbox{
$\sigma^{(j)}_s({\bs m})=\sigma^{(j)}_0({\bs m})$ for all $j=\{1. \dots, N\}$ and $s\in [0, \infty)$
}
\Big\}.
\end{align}
Here,  in the definition of $D_{\mathrm{S}}$, $\sigma_s^{(j)}(\bs m)$ denotes the flavor part  of $X_s^{(j)}({\bs m})$.
 Note that, 
  for a given ${\bs m}\in D$, the path $({\bs X}_t({\bs m}))_{t\ge 0}$ has certain characteristics. First, the flavor components $\sigma_t^{(j)}({\bs m})$ are constant in time, meaning that the flavor of each fermion remains unchanged throughout the trajectory. Second, fermions of equal flavor never meet each other, indicating that fermions with the same flavor do not occupy the same spatial position at any given time. 

By using the Feynman--Kac--It\^o  formula for a single fermion \eqref{FKI1} and Trotter's product formula, one obtains the following:

\begin{Prop}\label{NFermi}
 For every   $
 {\bs X}=\big(X^{(1)},\dots, X^{(N)}\big)\in \Omega_{\neq}^N
$
and $F\in \ell^2_{\rm as}(\Omega^N)$,  we have
\begin{align}
\Big(e^{-t H_{\Lambda}^{\rm H}({\bs b})} F
\Big)({\bs X})
=\Ex_{\bs X}\Bigg[ \mathbbm{1}_{D}
 \exp\Bigg\{ 
 -\int_0^{t} W({\bs X}_s)ds
 \Bigg\}F({\bs X}_{t})
\Bigg],  \label{LFK}
\end{align}
 where $\Ex_{\bs X}[F]$ represents the expected value of $F$ associated with the probability measure $\otimes_{j=1}^N P_{X^{(j)}}$ on $(\Omega^N, M^N)$, and 
 \be
 W({\bs X})=V({\bs X})+\sum_{j=1}^N v(X^{(j)}),\ \ {\bs X}\in  \Omega_{\neq}^N.
 \ee
\end{Prop}

\subsubsection{The system of $U=\infty$}

Here, we present a Feynman--Kac--It{\^o} formula that elucidates the semigroup generated by the Hamiltonian $H_{\Lambda}({\bs b})$ in the context of an infinitely strong on-site Coulomb repulsion, i.e., $U=\infty$.

Let us define the set $\Omega_{\neq, \infty}^N$ as follows:
\begin{align}
\Omega_{\neq, \infty}^N=\Big\{
{\bs X}\in \Omega^N\, \Big|\, x^{(i)} \neq x^{(j)}\ \ \text{for all $i, j\in {1, \dots, N}$ with $i\neq j$}
\Big\}.
\end{align}
In the aforementioned definition, we employ the following notations: ${\bs X}=(X^{(1)}, \dots, X^{(N)})
$ with $X^{(j)}=(x^{(j)}, \sigma^{(j)})$.
Throughout the remainder of this paper, we shall make use of the subsequent natural correspondence:
\be
\F_N=\ell_{\rm as}^2(\Omega_{\neq, \infty}^N),
\ee
where $\F_N$ is precisely defined by \eqref{DefF_N}.

Given $\beta>0$, we define the event as 
$D_{\infty}(\beta)=D_{\mathrm{O}, \infty}(\beta) \cap  D$, where 
\begin{align}
D_{\mathrm{O},  \infty}(\beta)=\Big\{
{\bs m}\in (M)^N\, \Big|\, {\bs X}_s({\bs m})\in \Omega_{\neq, \infty}^N\, \mbox{ for all $s\in [0, \beta]$}
\Big\}.
\end{align}
It should be noted that for each ${\bs m} \in D_{\infty}(\beta)$, there are no encounters between fermions along the corresponding path $({\bs X}_t({\bs m}))_{t\in [0, \beta]}$.

Now we are ready to construct  a Feynman--Kac--It{\^o}   formula for  $e^{-\beta H_{\Lambda}({\bs b})}$.
\begin{Thm}\label{FKIHubbrdUInf}
For every   ${\bs X} \in \Omega_{\neq, \infty}^N$  and $F\in \ell^2_{\rm as}(\Omega_{\neq, \infty}^N)$,  we have
\begin{align}
\Big(e^{-\beta H_{\Lambda}({\bs b})} F
\Big)({\bs X})
=\Ex_{\bs X}\Bigg[ \mathbbm{1}_{D_{\infty}(\beta)}
 \exp\Bigg\{-\int_0^{\beta} W_{\infty}({\bs X}_s) ds
\Bigg\} F({\bs X}_t)
\Bigg], \label{InftyFKF}
\end{align}
where 
\be
W_{\infty}({\bs X})=V_{\rm o}({\bs X})+\sum_{j=1}^Nv(X^{(j)}),\ \ {\bs X}\in  \Omega_{\neq, \infty}^N.
\ee
Here, recall that $V_{\rm o}$ is given by \eqref{VasMulti}.
\end{Thm}
\begin{proof}
By Proposition \ref{EffH}, we have 
\begin{align}
\lim_{U\to \infty}  e^{-\beta H_{\Lambda}^{\rm H} ({\bs b})} F 
= e^{-\beta H_{\Lambda}({\bs b})} F. \label{UInfty}
\end{align}
We denote by   $\mathbbm{1}_D G_U({\bs X}_{\bullet})$ the integrand in the right hand side of \eqref{LFK}. 
We split 
$\Ex_{{\bs X}}[\mathbbm{1}_DG_U({\bs X}_{\bullet})]$
 into two parts as  follows:
\begin{align}
\Ex_{{\bs X}}[\mathbbm{1}_DG_U({\bs X}_{\bullet})]=\Ex_{{\bs X}}[\mathbbm{1}_{D_{\infty}(\beta) }G_{U=0}({\bs X}_{\bullet})]
+\Ex_{{\bs X}}[\mathbbm{1}_{D\setminus D_{\infty}(\beta)} G_{U=0}({\bs X}_{\bullet})].
\end{align}
Because $
\lim_{U\to \infty} G_U({\bs X}_{\bullet}({\bs m}))=0
$ for all ${\bs m} \in D\setminus D_{\infty}(\beta)$, we have
\begin{align}
\lim_{U\to \infty}\Ex_{{\bs X}}[\mathbbm{1}_DG_U({\bs X}_{\bullet})]=\Ex_{{\bs X}}[\mathbbm{1}_{D_{\infty}(\beta) }G_{U=0}({\bs X}_{\bullet})]\label{InftyUMain}
\end{align}
 by the dominated convergence theorem.
Combining  (\ref{UInfty}) and (\ref{InftyUMain}), we obtain the desired statement as presented in Theorem \ref{FKIHubbrdUInf}.
\end{proof}

\subsubsection{A Feynman--Kac--It{\^o}  formula for the partition function}
Given $\beta>0$, let us define the event $D_{\mathrm{P}, \infty}(\beta)$ as follows:
\begin{align}
D_{\mathrm{P}, \infty}(\beta)=\big\{{\bs m}\in (M)^N\, \big|\, \mbox{$\exists\tau \in \mathfrak{S}_{N}({\bs X}_0({\bs m}))$ such that ${\bs X}_{\beta}({\bs m})=\tau {\bs X}_0({\bs m})$} \big\}, \label{DefDP}
\end{align}
Here, $\mathfrak{S}_N({\bs X})$ is a subset of $\mathfrak{S}_N$, and its precise definition will be provided in Definition \ref{DynP} below due to its somewhat intricate nature.
 We then define the event as:
\begin{align}
L_{\beta}=D_{\infty}(\beta) \cap D_{\mathrm{P}, \infty}(\beta).
\end{align}

The purpose here is to prove the following theorem:
\begin{Thm}\label{FKIUINFINITY}
For every $\beta >0$, there exists  a measeure $\mu_{\beta}$ on $L_{\beta}$ such that the following equation holds:
\begin{align}
Z_{\Lambda}(\beta; {\bs b}) 
=
\int_{L_{\beta}}d \mu_{\beta}
\prod_{j=1}^{N}\prod_{\sigma=1}^{n-1}\exp\Bigg\{\int_0^{\beta} b_{\sigma} k_{\sigma}(X^{(j)}_s)ds\Bigg\},
\label{FKITrUInf}
\end{align}
where $k_{\sigma}$ is given by \eqref{Defk}.
\end{Thm}

To prove Theorem \ref{FKIUINFINITY}, we need to make some preparations.
First, we will construct a complete orthonormal system (CONS) for $\F_N=\ell^2_{\mathrm{as}}(\Omega_{\neq, \infty}^N)$.
Given ${\bs X} \in \Omega^N$, we define
\be
\delta_{\bs X}=\otimes_{j=1}^N\delta_{X^{(j)}} \in \ell^2(\Omega^N),\quad e_{{\bs X}} =A_N \delta_{{\bs X}}\in\ell_{\mathrm{as}}^2(\Omega_{\neq, \infty}^N) ,
\ee
where $A_N$ is the antisymmetrizer defined by Eq.   \eqref{DefAN}. It can be readily verified that $\{\delta_{{\bs X}}\, |\, {\bs X} \in \Omega^N\}$ forms a CONS for $\ell^2(\Omega^N)$.
In order to construct a CONS for $\ell_{\mathrm{as}}^2(\Omega_{\neq, \infty}^N)$, we need to make further preparations. We observe that $e_{\tau {\bs X}}=\mathrm{sgn}(\tau) e_{{\bs X}}$ holds for all $\tau\in \mathfrak{S}_N$. With this in mind, we introduce an equivalence relation in $\Omega_{\neq, \infty}^N$ as follows:
Given ${\bs X}, {\bs Y} \in \Omega_{\neq, \infty}^N$, if there exists $\tau\in \mathfrak{S}_N$ such that ${\bs Y}=\tau {\bs X}$, then we write ${\bs X} \equiv {\bs Y}$.
It can be easily verified that this binary relation defines an equivalence relation.
Let $[{\bs X}]$ denote the equivalence class to which ${\bs X}$ belongs. We will often abbreviate $[{\bs X}]$ as ${\bs X}$ when there is no ambiguity.
We denote $[\Omega_{\neq, \infty}^N]$ as the quotient set $\Omega_{\neq, \infty}^N/\equiv$.
It can be shown that  $\{e_{\bs X}\, |\, {\bs X} \in [{\Omega_{\neq, \infty}^N}]\}$ forms a CONS for $\ell^2_{\mathrm{as}}(\Omega_{\neq, \infty}^N)$. This CONS will be useful in our subsequent analysis.

Let $Q_{\Lambda}$ be the orthogonal projection from $\ell^2(\Omega^N)$ to $\F_N$ as defined in  \eqref{DefQ}.
We denote ${\bs T}$ as the Hamiltonian for the free fermions:
\be
{\bs T}=Q_{\Lambda} T_{{\bs b} ={\bs 0}, \mu=0}Q_{\Lambda},
\ee
where $T_{{\bs b}={\bs 0}, \mu=0}$ represents the operator obtained by setting ${\bs b}={\bs 0}$ and $\mu=0$ in the defining equation of $T$, i.e., Eq.  \eqref{NELHAMI}.

We are now ready to state the precise definition of  $\mathfrak{S}_N({\bs X})$  that appears in Eq. \eqref{DefDP}:
\begin{Def}\label{DynP}\rm

Let ${\bs X}\in \Omega^N_{\neq, \infty}$.
We define a permutation $\tau\in \mathfrak{S}_N$ to be {\it dynamically allowed} associated with ${\bs X}$ if there exists an $n\in\BbbZ+$ such that
\begin{align}
\la \delta_{{\bs X}}|{\bs T}^n \delta_{\tau{\bs X}}\ra \neq 0. \label{DyAlll1}
\end{align}
We denote the set of all dynamically allowed permutations associated with ${\bs X}$ as $\mathfrak{S}_N({\bs X})$.
It is worth noting that if $\tau$ is dynamically allowed, it is always an even permutation, i.e., $\mathrm{sgn}(\tau)=1$ \cite{Aizenman1990}.

\end{Def}

To provide a characterization of dynamically allowed permutations, we introduce several terms.
Consider ${\bs X}=(X^{(j)})_{j=1}^N\in \Omega^N_{\neq, \infty}$ and $ {\bs Y}=(Y^{(j)})_{j=1}^N\in \Omega^N_{\neq, \infty}$.
We define the distance between ${\bs X}$ and ${\bs Y}$ as
\be
\|{\bs X}-{\bs Y}\|_{\infty}=\max_{j=1, \dots, N}\|x^{(j)}-y^{(j)}\|_{\infty},\label{MaxNorm}
\ee 
where $x^{(j)}$ (resp. $y^{(j)}$) represents the spatial component of $X^{(j)}$ (resp. $Y^{(j)}$).
Here, the norm $\|\cdot \|_{\infty}$ present on the right-hand side of Eq. \eqref{MaxNorm} corresponds to the maximum norm defined over the set $\Lambda$
: $\|x\|_{\infty}=\max_{j=1, \dots, n} |x_j|\ (x=(x_j)\in \Lambda)$.
We say that ${\bs X}$ and ${\bs Y}$ are {\it neighbors} if $\|{\bs X}-{\bs Y}\|_{\infty}=1$ and the flavor components of $X^{(j)}$ and $Y^{(j)}$ are equal for all $j=1, \dots, n$.
An {\it edge} is a pair $\{ {\bs X}, {\bs Y}\} \in \Omega^N_{\neq, \infty} \times \Omega^N_{\neq, \infty}$ where ${\bs X}$ and ${\bs Y}$ are neighbors.
A {\it path} is a sequence $({\bs X}_i)_{i=1}^m \subset \Omega_{\neq, \infty}^N$ such that $\{{\bs X}_i, {\bs X_{i+1}}\}$ is an edge for all $i$.
For a given edge $\{{\bs X}, {\bs Y}\}$, we define the linear operator $Q({\bs X}, {\bs Y})$ acting on $\ell^2(\Omega_{\neq, \infty}^N)$ as
\begin{align}
Q({\bs X}, {\bs Y})= |\delta_{{\bs X}}\ra \la \delta_{{\bs Y}}|.
\end{align}
This operator plays a role in describing the following lemma, which provides a characterization of dynamically allowed permutations.

\begin{Lemm}\label{EquivPath}
Let $\tau\in \mathfrak{S}_N$ and let ${\bs X}\in\Omega^N_{\neq, \infty}$.
The following {\rm (i)} and {\rm (ii)} are mutually equivalent:
 \begin{itemize}
 \item[{\rm (i)}] $\tau$ is dynamically allowed associated with  ${\bs X}$;
 \item[{\rm (ii)}] there exists a path $({\bs X}_i)_{i=1}^m$ satisfying the following:
 \begin{itemize}
 \item[] ${\bs X}_1={\bs X}$ and ${\bs X}_m=\tau {\bs X}$;
 \item[] $
 \la \delta_{{\bs X}}|Q({\bs X}_1, {\bs X}_2) Q({\bs X}_2, {\bs X}_3) \cdots Q({\bs X}_{m-1}, {\bs X}_m) \delta_{\tau {\bs X}}\ra >0.
 $
 \end{itemize}
 \end{itemize}
 \end{Lemm}
See \cite{Miyao2020-2} for a proof of this lemma.

We are now ready to prove Theorem \ref{FKIUINFINITY}.
\begin{proof}[Proof of Theorem \ref{FKIUINFINITY}]
We divide the proof into two parts.

{\bf Step 1.}
Consider $\ell^{\infty}_{\rm s}(\Omega^N)$ as the Banach space comprising all symmetric functions defined on $\Omega^N$, endowed with the infinite norm.
Let $F_0, F_1, \dots, F_{n-1}$ denote strictly positive elements in $\ell_{\rm s}^{\infty}(\Omega^N)$.
We define the operator $K_n$ as follows:
  \be
K_n
=F_0 e^{-t_1H_{\Lambda}({\bs b})} F_1 e^{-(t_2-t_1)H_{\Lambda}({\bs b})}F_2 \cdots F_{n-1}e^{-(\beta-t_{n-1}) H_{\Lambda}({\bs b})} .
\ee
Let ${\bs X}\in \Omega_{\neq, \infty}^N$ be fixed arbitrarily.
We claim  that if $\tau$ is {\it not} dynamically allowed associated with ${\bs X}$, then one obtains, for all $n\in \BbbN$ and $
0<t_1<t_2<\cdots <t_{n-1}<\beta
$, that 
\begin{align}
\la \delta_{{\bs X}}|K_{n} \delta_{\tau{\bs X}}\ra=0. \label{DAConc}
\end{align}
In this step, we will demonstrate the proof of this equation in a step-by-step manner.

 Let us first consider the case where $W_{\infty}\equiv0$. For the sake of simplicity, let us  assume that $n=2$.
  Because $F_0$ and $F_1$ are multipication operators, and $\tau$ is not dynamically allowed, we can deduce that 
  \begin{align}
  \la \delta_{{\bs X}}|F_0 (-{\bs T})^{n_1} F_1 (-{\bs T})^{n_2}\delta_{\tau{\bs X}}\ra=0 \label{DyAll2}
  \end{align}
    for all $n_1, n_2\in \BbbZ_+$.
    To establish this, it is worth noting that we can express ${\bs T}$ as: 
 \begin{align}
 {\bs T}=\sum_{\{{\bs X}, {\bs Y}\}} C_{{\bs X}, {\bs Y}}Q({\bs X}, {\bs Y})
+\mathcal{D}, \label{OfDD}
 \end{align}
 where 
 $\sum_{\{{\bs X}, {\bs Y}\}}$ denotes the summation over all edges,  the coefficients $C_{{\bs X}, {\bs Y}}$ satisfy
 $C_{{\bs X}, {\bs Y}}<0
 $  for each edge $\{{\bs X}, {\bs Y}\}$,  and $\mathcal{D}$ represents a multiplication operator.
 By  using the formula \eqref{OfDD}, 
    Eq.  (\ref{DyAll2}) follows from the following property:
    \begin{align}
    \la \delta_{{\bs X}}|Q({\bs X}_1, {\bs X}_2) Q({\bs X}_2, {\bs X}_3) \cdots Q({\bs X}_{m-1}, {\bs X}_m) \delta_{\tau {\bs X}}\ra =0
    \end{align}
    for any path $({\bs X}_i)_{i=1}^m$. This property is evident from Lemma \ref{EquivPath}.
   Using (\ref{DyAll2}), we can prove (\ref{DAConc}) as follows:
   \begin{align}  
     \la \delta_{{\bs X}}|K_{n} \delta_{\tau{\bs X}}\ra 
     =\sum_{n_1=1}^{\infty}\sum_{n_2=1}^{\infty}\frac{t_1^{n_1} (t_2-t_1)^{n_2}}{n_1!n_2!}\la \delta_{{\bs X}}|F_0 (-{\bs T})^{n_1} F_1 (-{\bs T})^{n_2}\delta_{\tau{\bs X}}\ra=0.
      \end{align}
Similarly, we can prove \eqref{DAConc} for general $n$, when $W_{\infty}\equiv 0$.

Next, let us consider the case where $W_{\infty} \neq 0$. 
Once again, we will focus on the case $n=2$ for simplicity. By employing Trotter's formula, we obtain
\begin{align}
&\la \delta_{{\bs X}}|K_{n} \delta_{\tau{\bs X}}\ra\no
=&\lim_{N_1\to \infty}\lim_{N_2\to \infty}
\big\la \delta_{\bs X}| F_0\big(
e^{-t_1 {\bs T}/N_1} e^{-t_1W_{\infty}/N_1}
\big)^{N_1 }  F_1\big (
e^{-(t_2-t_1) {\bs T}/N_2} e^{-(t_2-t_1)W_{\infty}/N_2}
\big)^{N_1 }\delta_{\tau {\bs X}}\big\ra. \label{TKN2}
\end{align}
 By applying  the claim for the case where  $W_{\infty} \equiv 0$,  we observe that the right-hand side of (\ref{TKN2}) evaluates to zero.
 Similarly, we can establish the assertion for general $n$.
Thus, we have successfully completed the proof of Eq. \eqref{DAConc}.

{\bf Step 2.}
By using Theorem \ref{FKIHubbrdUInf} and (\ref{DAConc}), we get 
\begin{align}
Z_{\Lambda}(\beta; {\bs b})
=&\sum_{{\bs X} \in [\Omega_{\neq, \infty}^N]} \la e_{{\bs X}}\, |\, e^{-\beta H_{\Lambda}({\bs b})}e_{{\bs X}}\ra\no
=& \sum_{{\bs X} \in [\Omega_{\neq, \infty}^N]} \sum_{\tau \in \mathfrak{S}_N} \frac{\mathrm{sgn}(\tau)}{N!}
\la \delta_{{\bs X}}\, |\,  e^{-\beta H_{\Lambda}({\bs b})} \delta_{\tau {\bs X}}\ra\no
=& \sum_{{\bs X} \in [\Omega_{\neq, \infty}^N]} \sum_{\tau \in \mathfrak{S}_N({\bs X})} \frac{\mathrm{sgn}(\tau)}{N!}
\la \delta_{{\bs X}}\, |\, e^{-\beta H_{\Lambda}({\bs b})} \delta_{\tau {\bs X}}\ra\no
=& \sum_{{\bs X} \in [\Omega_{\neq, \infty}^N]} \sum_{\tau \in \mathfrak{S}_N({\bs X})} \frac{\mathrm{sgn}(\tau)}{N!}
\Ex_{{\bs X}} \Big[
\mathbbm{1}_{\{{\bs X}_{\beta}=\tau{\bs X}\}\cap D_{\infty}(\beta)} e^{-\int_0^{\beta} W_{\infty}({\bs X}_s) ds}
\Big]. \label{TraF1}
\end{align}
Let us proceed by defining the measure on $L_{\beta}$ as follows:
\begin{align}
\mu_{\beta}(B)=\sum_{{\bs X} \in [\Omega_{\neq, \infty}^N]} \sum_{\tau\in \mathfrak{S}_N({\bs X})} \frac{1}{N!} 
\Ex_{\bs X}\Bigg[
\mathbbm{1}_B \mathbbm{1}_{\{{\bs X}_{\beta}=\tau{\bs X}\} \cap D_{\infty}(\beta)}
e^{-\int_0^{\beta} W_{\infty, {\bs b}={\bs 0}}({\bs X}_s) ds}
\Bigg], \label{Measurenu}
\end{align}
where $W_{\infty, {\bs b}={\bs 0}}({\bs X})$ is obtained by substituting ${\bs b}={\bs 0}$ into the defining equation of $W_{\infty}({\bs X})$.
Because $\mathrm{sgn}(\tau)=1$ for all $\tau\in \mathfrak{S}_N({\bs X})$, we finally  obtain the desired assertion stated in  Theorem \ref{FKIUINFINITY}.
\end{proof}

\section{Proofs of Theorems \ref{Main1} and \ref{Main2}} \label{Sec4}

\subsection{Random loop representations}
For a given ${\bs m}\in L_{\beta}$, we define the particle world lines  associated with  the path $({\bs X}_s({\bs m}))_{s\in [0, \beta]}$ as follows: 
\be
{\bs x}_t({\bs m})=(x_t^{(1)}({\bs m}), \dots, x_t^{(N)}({\bs m})),
\ee
 where $x^{(i)}_t({\bs m})$ denotes the spatial component of $X_t^{(i)}({\bs m})$.
Given that each $x_t^{(i)}$ takes values on $\Lambda$, it maintains a piecewise constant behavior with respect to time and undergoes transitions to nearest neighbor sites at random times. As a result, the particle world lines can be visually represented as a collection of polygonal lines in the space-time picture.

\begin{figure}[t]
\centering
 \includegraphics[scale=0.5]{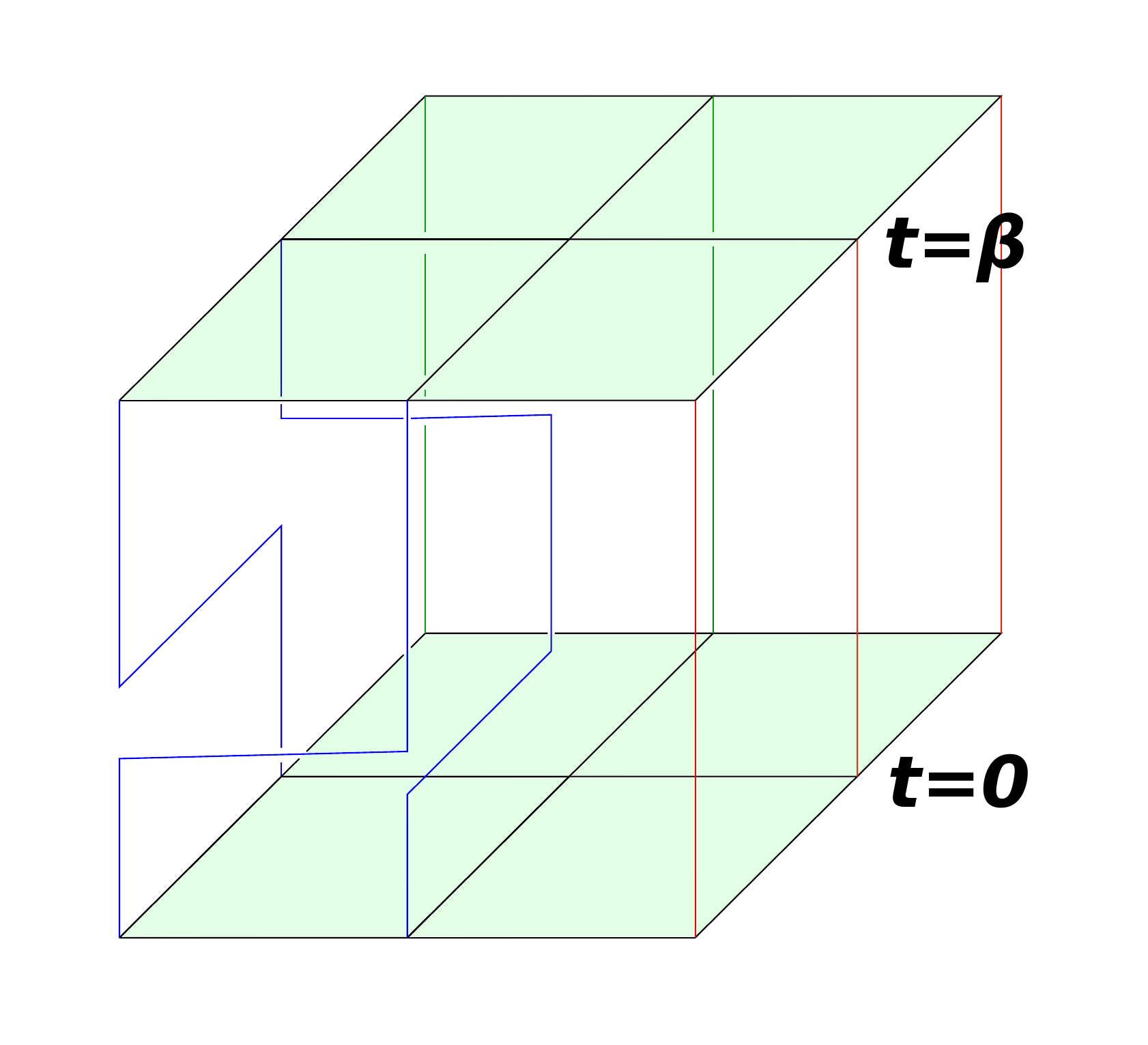}
  \caption{
The figure illustrates loops formed by the world lines of eight particles on the lattice $\Lambda=\{-1, 0, 1\}\times \{-1, 0, 1\}$. The colors assigned to the world lines indicate the flavors of the fermions.
        }\label{Fig21}
 \end{figure}

 In accordance with the approach outlined in  \cite{Aizenman1990}, we can establish a correspondence between arbitrary particle world lines and a collection of loops within the space-time picture through the following steps:
 \begin{itemize}
\item We initiate the process by plotting the loops starting from the positions of each particle at $t=0$.
\item Subsequently, we trace the temporal evolution of a particle's spatial position in the space-time picture. It is important to note that, due to the infinite strength of the Coulomb interaction, the particles never come into contact with one another.
\item Once the trace line reaches the time $t=\beta$, it reemerges in the same location at $t=0$ by considering time as periodic.
\item We continue these steps until the trace line forms a closed loop.
\end{itemize}

Every path can be fully characterized by the following three conditions:
the initial configuration of the particles: ${\bs x}_0$, the particle world lines, and the flavor assigned along the world lines.
Hence, it is possible to uniquely specify each path by assigning a flavor to each loop.
 Figure \ref{Fig21} illustrates typical loops composed of particle world lines in a two-dimensional system.

Let us consider the set of loops $\{\ell_1, \dots, \ell_k\}$ associated with a given path $({\bs X}_s({\bs m}))_{s\in [0, \beta]}$.
 In general, a combination of a loop $\ell$ and a flavor $\sigma$, denoted as $\gamma=(\ell, \sigma)$, is referred to as a {\it flavored loop}. 
Consequently, the path $({\bs X}_s({\bs m}))_{s\in [0, \beta]}$ can be uniquely identified by the collection of flavored loops $\varGamma({\bs m})=\{\gamma_1, \dots, \gamma_k\}$, where $\gamma_j=(\ell_j, \sigma_j)$.
This set $\varGamma({\bs m})$ is termed the {\it flavored random loop}.

By using Theorem \ref{FKIUINFINITY}, we obtain the following:

\begin{Thm}
The partition function has the following random loop representation:
\begin{align}
Z_{\Lambda}(\beta;  {\bs b})
=\int_{L_{\beta}} d\mu_{\beta} 
\prod_{\gamma\in \varGamma}
\exp\{ \beta w_{\gamma} f(\sigma_{\gamma})\}, 
\end{align}
where, for each $\gamma=(\ell_{\gamma}, \sigma_{\gamma})\in \varGamma$,
$w_{\gamma}$ denotes the absolute value of
the  winding number of the loop $\ell_{\gamma}$, and the function $f$ on $\{1, \dots, n\}$ is given by 
\begin{align}
 f(\sigma)
 =\begin{cases}
 b_1 & (\sigma=1)\\
 -b_{\sigma-1}+b_{\sigma}  & (2\le \sigma\le n-1)\\
 - b_{n-1} & (\sigma=n).
 \end{cases}
\end{align}
\end{Thm}
\begin{proof}
For a given ${\bs m}\in L_{\beta}$, let us consider the flavored loops $\varGamma(\bs m)=\{\gamma_1, \dots, \gamma_k\}$ corresponding to the path $(\bs X_s({\bs m}))_{s}$.
Expressing each flavored loop as $\gamma_j=(\ell_j, \sigma_j)$, we can represent each loop $\ell_j$ in terms of the particle world lines as follows: $\ell_j=\bigcup_{s\in [0, \beta]}\bigcup_{i\in I_j} x_s^{(i)}(\bs m)$,
where $I_j$ denotes the set of particle labels that compose $\ell_j$.
Using these symbols, we get
\begin{align}
\prod_{j=1}^{N}\prod_{\sigma=1}^{n-1}\exp\Bigg\{\int_0^{\beta} b_{\sigma} k_{\sigma}(X^{(j)}_s({\bs m}))ds\Bigg\}
&=\prod_{j=1}^k \underbrace{\prod_{\sigma=1}^{n-1} \prod_{i\in I_j} \exp\Bigg\{\int_0^{\beta} b_{\sigma} k_{\sigma}(X^{(i)}_s({\bs m}))ds\Bigg\}}_{=: \mathscr{C}(\gamma_j)}\no
&=\prod_{j=1}^k\mathscr{C}(\gamma_j).
\end{align}
Recalling \eqref{Defk}, we obtain
\begin{align}
\sum_{i\in I_j} \sum_{\sigma=1}^{n-1} b_{\sigma} k_{\sigma}(X_s^{(i)}({\bs m}))
&=\sum_{i\in I_j} \sum_{x\in \Lambda}\sum_{\sigma=1}^{n-1} b_{\sigma}\Big\{
\delta_{(x, \sigma)} (X_s^{(i)}({\bs m}))-\delta_{(x, \sigma+1)} (X_s^{(i)}({\bs m}))
\Big\}\no
&=\sum_{i\in I_j} \sum_{x\in \Lambda}\delta_{x, x_s^{(i)}({\bs m})} f(\sigma_{j})=|I_j| f(\sigma_{j}).
\end{align}
Combining this with the fact $|I_j|=w(\gamma_j)$, we conclude  that $\mathscr{C}(\gamma)=e^{\beta w(\gamma) f(\sigma_{\gamma})}$.
\end{proof}

\subsection{Proof of Theorem \ref{Main1}}
It should be noted that for a given set of flavored random loops $\varGamma({\bs m})=\{\gamma_1, \dots, \gamma_k\}$, the corresponding collection of winding numbers ${\bs w}_{\varGamma({\bs m })}=(w_{\gamma_1}, \dots, w_{\gamma_k})$ forms a partition of $N$:
\be
\sum_{j=1}^k w_{\gamma_j}=N.
\ee
Hence, as ${\bs m}$ varies within $L_{\beta}$, ${\bs w}_{\varGamma({\bs m})}$ traverses various partitions of $N$.

To state a technical lemma, we need to introduce some symbols.
Let $L$ denote the collection of loops associated with $\varGamma(\bs m)$, and let $F$ represent the collection of flavors associated with $\varGamma(\bs m)$:
\be
L(\varGamma(\bs m))=(\ell_{\gamma} :  \gamma\in \varGamma(\bs m)),
\ \ 
F(\varGamma(\bs m))=(\sigma_{\gamma} : \gamma\in \varGamma(\bs m)).
\ee
Note that, if $\# \varGamma({\bs m})=k$, then 
$F(\varGamma({\bs m}))$ belongs to $\mathscr{F}_k:=\{1, \dots, n\}^k$.
Here, the notation $\# S$ represents the cardinality of a given set $S$.
Given  a collection of flavors $F\in \mathscr{F}_k$, we define
\be
S(k; F)=\{ {\bs m} \in L_{\beta} : \#\varGamma(\bs m)=k,\ \ F(\varGamma({\bs m}))=F\}.
\ee
Next, we partition the set of partitions of $N$, denoted as $P_N$, in the following manner:
\be
P_N=\bigsqcup_{k=1}^NP_N(k),
\ee
where we define $P_N(k)=\{ {\bs n} \in P_N : \#{\bs n}=k\}$.
For each $F\in \mathscr{F}_k$ and ${\bs n}\in P_N(k)$, we define 
\be
S(k; F; {\bs n})=\{{\bs m} \in S(k; F) : {\bs w}_{\varGamma}={\bs n}\},
\ee
where ${\bs w}_{\varGamma}$ represents the collection of the winding numbers of the loops associated with $\varGamma$.
It is crucial for the proof of Theorem \ref{Main1} to note that $S(k; F)$ can be partitioned as follows:
\be
S(k; F)=\bigsqcup_{{\bs n} \in P_N(k)} S(k; F; {\bs n}).
\ee

\begin{Lemm}\label{Key}
Fix 
$k\in \{1, \dots, n\}$, arbitrarily. 
We also fix ${\bs n}\in P_N(k)$.
Then $\mu_{\beta}( S(k; F; {\bs n}))$  is independent of $F$ and constant on $\mathscr{F}_k$.
Setting
\be
D_{\beta}(\bs n)=\mu_{\beta}( S(k; F; {\bs n})), 
\ee
we obtain  the following identity:
\be
\sum_{F\in \mathscr{F}_k}\int_{S(k; F; {\bs n})} d\mu_{\beta} 
\prod_{\gamma\in \varGamma}
\exp\{ \beta w_{\gamma} f(\sigma_{\gamma})\}=
D_{\beta}({\bs n}) \mathcal{G}_{\beta}({\bs n}; {\bs b}),
\ee
where $ \mathcal{G}_{\beta}({\bs n}; {\bs b})$ is given by \eqref{DefG1}.

\end{Lemm}
\begin{proof}
Let $\mathfrak{S}_n$ be the permutation group on the set $\{1, 2, \dots, n\}$.
For any ${\bs\vepsilon}=(\vepsilon_j)_{j=1}^N\in \mathfrak{S}_n^N$, we define the unitary operator $U_{\bs \vepsilon}$ on $\F_N$ by
\be
(U_{\bs \vepsilon}\vphi)({\bs X})=\vphi({\bs X}_{\bs \vepsilon}),\ \ \vphi\in \F_N,
\ee
where 
${\bs X_{\bs \vepsilon}}=(X^{(1)}_{\bs \vepsilon}, \dots, X_{\bs \vepsilon}^{(N)})$ is defined by 
$X^{(j)}_{\bs \vepsilon}=(x^{(j)}, \vepsilon_j(\sigma^{(j)}))$.
Because 
\be
U_{\bs \vepsilon} H_{\Lambda}({\bs 0}) U_{\bs \vepsilon}^{-1} =H_{\Lambda}({\bs 0})
\ee
holds for any ${\bs \vepsilon}\in \mathfrak{S}_n^N$,
we readily confirm that $\mu_{\beta}( S(k; F; {\bs n}))$  is independent of $F$ and constant on $\mathscr{F}_k$.

Since the function $\prod_{\gamma\in \varGamma}
\exp\{ \beta w_{\gamma} f(\sigma_{\gamma})\}$ is constant on 
$S(k; F; {\bs n})$, we see that 
\be
\int_{S(k; F; {\bs n})} d\mu_{\beta}  \prod_{\gamma\in \varGamma}
\exp\{ \beta w_{\gamma} f(\sigma_{\gamma})\}
=\mu_{\beta}(S(k; F; {\bs n}))\prod_{j=1}^k \exp\{ \beta n_j f(\sigma_j)\},
\ee
where we set $F=(\sigma_1, \dots, \sigma_k)$ and ${\bs n}=(n_1, \dots, n_k)$.
Therefore, by combining the first half of the statement with this fact, we obtain  the following:
\begin{align}
\sum_{F\in \mathscr{F}_k}\int_{S(k; F; {\bs n})} d\mu_{\beta} 
\prod_{\gamma\in \varGamma}
\exp\{ \beta w_{\gamma} f(\sigma_{\gamma})\}&=
\sum_{F\in \mathscr{F}_k}D_{\beta}({\bs n})
\prod_{j=1}^k
\exp\{ \beta n_j f(\sigma_j)\}\no
&=D_{\beta}({\bs n}) \prod_{j=1}^k \sum_{\sigma=1}^n\exp\{ \beta n_j f(\sigma)\}\no
&=D_{\beta}({\bs n}) \mathcal{G}_{\beta}({\bs n}; {\bs b}).
\end{align}
We are now done with the proof of Lemma \ref{Key}.
\end{proof}

\begin{proof}[Completion of the proof of Theorem \ref{Main1}]
Using the fact that $L_{\beta}$ can be divided as
\be
L_{\beta}
=\bigsqcup_{k=1}^N \bigsqcup_{F\in \mathscr{F}_k} S(k; F)=
\bigsqcup_{k=1}^N \bigsqcup_{F\in \mathscr{F}_k} \bigsqcup_{{\bs n} \in P_N(k)} S(k; F; {\bs n}), 
\ee
  the partition function can be expressed as follows:
\begin{align}
Z_{\Lambda}(\beta;  {\bs b}) =&\sum_{k=1}^N \sum_{{\bs n}\in P_N(k)}\sum_{F\in \mathscr{F}_k} \int_{S(k; F; {\bs n})}  d\mu_{\beta} 
\prod_{\gamma\in \varGamma}
\exp\{ \beta w_{\gamma} f(\sigma_{\gamma})\}\no
=&\sum_{k=1}^N \sum_{{\bs n} \in P_N(k)} D_{\beta}({\bs n}) \mathcal{G}_{\beta}({\bs n}; {\bs b})\no
=& \mbox{the RHS of \eqref{ExZ}}.
\end{align}
This completes the proof of Theorem \ref{Main1}.
\end{proof}

\subsection{Proof of Theorem \ref{Main2}}

First, let us observe that the expected value $\la h_{\sigma}\ra$ can be expressed as follows:
\begin{align}
\beta \la h_{\sigma} \ra_{\beta}&=\frac{\partial}{\partial b_{\sigma}} \log Z_{\Lambda}(\beta; \bs b)\no
&=Z_{\Lambda}(\beta; {\bs b})^{-1} \sum_{{\bs n} \in P_N}\sum_{j=1}^k D_{\beta}(\bs n)\mathcal{G}_{\beta}({\bs n}; {\bs b})\frac{\frac{\partial }{\partial b_{\sigma}} G_{\beta}(n_j;{\bs b})}{ G_{\beta}(n_j; {\bs b})}. \label{ExpVG}
\end{align}
Next, we aim to estimate the quantity  
$\frac{\partial}{\partial b_{\sigma}} G_{\beta}(m; {\bs b}) \Big/G_{\beta}(m; {\bs b})$
 from below.
For this purpose, we introduce the following definitions:
\begin{align}
C_{\sigma}(m) &=e^{\beta m B_{\sigma}}+e^{\beta m B_{\sigma+1}} ,\\
S_{\sigma}(m) &=e^{\beta m B_{\sigma}}-e^{\beta m B_{\sigma+1}} .
\end{align}
Then  $G_{\beta}(m; {\bs b})$ can be expressed as 
\be
G_{\beta}(m; {\bs b})=\sum_{\sigma=1}^{n-1} e^{\beta m B_{\sigma}}
=C_{\sigma}(m)+H_{\sigma}(m),
\ee
where
\be
H_{\sigma}=\sum_{\tau\neq \sigma, \sigma+1} e^{\beta m B_{\tau}}.
\ee
From this representation, the equality
\be
\frac{\partial}{\partial b_{\sigma}} G_{\beta}(m; {\bs b})=\beta m S_{\sigma}(m)
\ee
 follows immediately. Hence, we obtain
\be
\frac{\partial}{\partial b_{\sigma}} G_{\beta}(m; {\bs b}) \Big/G_{\beta}(m; {\bs b})=\beta m\frac{S_{\sigma}/C_{\sigma}}{1+H_{\sigma}/C_{\sigma}}.
\ee
If we define the function $f$ by 
$f(x)=\frac{1-e^{-x}}{1+e^{-x}}$, then since $f$ is  monotonically increasing, we have the following inequality:
\be
\frac{S_{\sigma}(m)}{C_{\sigma}(m)} =f(\beta m (B_{\sigma}-B_{\sigma+1})) \ge f(\beta(B_{\sigma}-B_{\sigma+1}))=f_{\beta, \sigma}(\bs b), 
\ee
where we use the fact that $B_{\sigma}>B_{\sigma+1}$ and $f_{\beta, \sigma}({\bs b})$ is  defined by \eqref{Deff}.
A quick examination, on the other hand, also reveals the following inequality:
\begin{align}
\frac{H_{\sigma}(m)}{C_{\sigma}(m)}&=\frac{\sum_{\tau\neq \sigma, \sigma+1} e^{\beta m (B_{\tau}-B_{\sigma})}}{1+e^{-\beta m (B_{\sigma}-B_{\sigma+1})}}\no
&\le \sum_{\tau\neq \sigma, \sigma+1} e^{\beta m (B_{\tau}-B_{\sigma})}\le g_{\beta, \sigma}(\bs b),
\end{align}
where $g_{\beta, \sigma}(\bs b)$ is given by \eqref{Defg}. Here, we used the assumption that $B_{\sigma}>B_{\tau}\ (\tau\neq \sigma)$  in deriving the second inequality.
Putting the above inequalities together, we get
\begin{align}
\frac{\partial}{\partial b_{\sigma}} G_{\beta}(m; {\bs b}) \Big/G_{\beta}(m; {\bs b})\ge \beta m \frac{f_{\beta, \sigma}(\bs b)} {1+g_{\beta, \sigma}(\bs b)},
\end{align}
which implies that 
\begin{align}
\mbox{the RHS of \eqref{ExpVG} }\ge \beta \frac{ f_{\beta, \sigma}(\bs b)}{1+g_{\beta, \sigma}(\bs b)} N,
\end{align}
where we use the fact $\sum_{j=1}^k n_j=N$. 
We are now done with the proof of Theorem \ref{Main2}.
\qed


\begin{thebibliography}{10}

\bibitem{Aizenman1990}
M.~Aizenman and E.~H. Lieb.
\newblock {Magnetic properties of some itinerant-electron systems {at} $T>0$}.
\newblock {\em Physical Review Letters}, 65(12):1470--1473, Sept. 1990.
\newblock \href {https://doi.org/10.1103/physrevlett.65.1470}
  {\path{doi:10.1103/physrevlett.65.1470}}.

\bibitem{cmp/1104270709}
M.~Aizenman and B.~Nachtergaele.
\newblock {Geometric aspects of quantum spin states}.
\newblock {\em Communications in Mathematical Physics}, 164(1):17 -- 63, 1994.
\newblock URL: \url{https://doi.org/}, \href {https://doi.org/cmp/1104270709}
  {\path{doi:cmp/1104270709}}.

\bibitem{Cazalilla_2014}
M.~A. Cazalilla and A.~M. Rey.
\newblock Ultracold fermi gases with emergent {SU}($n$) symmetry.
\newblock {\em Reports on Progress in Physics}, 77(12):124401, nov 2014.
\newblock \href {https://doi.org/10.1088/0034-4885/77/12/124401}
  {\path{doi:10.1088/0034-4885/77/12/124401}}.

\bibitem{Gneysu2015}
B.~G\"{u}neysu, M.~Keller, and M.~Schmidt.
\newblock {A Feynman{\textendash}Kac{\textendash}It{\^{o}} formula for magnetic
  Schr\"{o}dinger operators on graphs}.
\newblock {\em Probability Theory and Related Fields}, 165(1-2):365--399, May
  2015.
\newblock \href {https://doi.org/10.1007/s00440-015-0633-9}
  {\path{doi:10.1007/s00440-015-0633-9}}.

\bibitem{Gutzwiller1963}
M.~C. Gutzwiller.
\newblock {Effect of Correlation on the Ferromagnetism of Transition Metals}.
\newblock {\em Physical Review Letters}, 10(5):159--162, Mar. 1963.
\newblock \href {https://doi.org/10.1103/physrevlett.10.159}
  {\path{doi:10.1103/physrevlett.10.159}}.

\bibitem{PhysRevX.6.021030}
C.~Hofrichter, L.~Riegger, F.~Scazza, M.~H\"ofer, D.~R. Fernandes, I.~Bloch,
  and S.~F\"olling.
\newblock {Direct Probing of the Mott Crossover in the $\mathrm{SU}(N)$
  Fermi-Hubbard Model}.
\newblock {\em Phys. Rev. X}, 6:021030, Jun 2016.
\newblock \href
  {https://doi.org/10.1103/PhysRevX.6.021030}
  {\path{doi:10.1103/PhysRevX.6.021030}}.

\bibitem{Hubbard1963}
J.~Hubbard.
\newblock {Electron correlations in narrow energy bands}.
\newblock {\em Proceedings of the Royal Society of London. Series A.
  Mathematical and Physical Sciences}, 276(1365):238--257, Nov. 1963.
\newblock \href {https://doi.org/10.1098/rspa.1963.0204}
  {\path{doi:10.1098/rspa.1963.0204}}.

\bibitem{Humphreys1972}
J.~E. Humphreys.
\newblock {\em Introduction to Lie Algebras and Representation Theory}.
\newblock Springer New York, 1972.
\newblock \href {https://doi.org/10.1007/978-1-4612-6398-2}
  {\path{doi:10.1007/978-1-4612-6398-2}}.

\bibitem{Kanamori1963}
J.~Kanamori.
\newblock {Electron Correlation and Ferromagnetism of Transition Metals}.
\newblock {\em Progress of Theoretical Physics}, 30(3):275--289, Sept. 1963.
\newblock \href {https://doi.org/10.1143/ptp.30.275}
  {\path{doi:10.1143/ptp.30.275}}.

\bibitem{Katsura2013}
H.~Katsura and A.~Tanaka.
\newblock {Nagaoka states in the SU($n$) Hubbard model}.
\newblock {\em Physical Review A}, 87(1), Jan. 2013.
\newblock \href {https://doi.org/10.1103/physreva.87.013617}
  {\path{doi:10.1103/physreva.87.013617}}.

\bibitem{Lieb2004}
E.~H. Lieb.
\newblock {The Hubbard model: Some Rigorous Results and Open Problems}.
\newblock In {\em Condensed Matter Physics and Exactly Soluble Models}, pages
  59--77. Springer Berlin Heidelberg, 2004.
\newblock \href {https://doi.org/10.1007/978-3-662-06390-3_4}
  {\path{doi:10.1007/978-3-662-06390-3_4}}.

\bibitem{LIU20191490}
R.~Liu, W.~Nie, and W.~Zhang.
\newblock {Flat-band ferromagnetism of SU($N$) Hubbard model on Tasaki lattices}.
\newblock {\em Science Bulletin}, 64(20):1490--1495, 2019.
\newblock 
  \href {https://doi.org/https://doi.org/10.1016/j.scib.2019.08.013}
  {\path{doi:https://doi.org/10.1016/j.scib.2019.08.013}}.

\bibitem{Miyao2017}
T.~Miyao.
\newblock {Nagaoka's Theorem in the Holstein{\textendash}Hubbard Model}.
\newblock {\em Annales Henri Poincar{\'{e}}}, 18(9):2849--2871, Apr. 2017.
\newblock \href {https://doi.org/10.1007/s00023-017-0584-z}
  {\path{doi:10.1007/s00023-017-0584-z}}.

\bibitem{Miyao2020-2}
T.~Miyao.
\newblock {Thermal Stability of the Nagaoka{\textendash}Thouless Theorems}.
\newblock {\em Annales Henri Poincar{\'{e}}}, 21(12):4027--4072, Oct. 2020.
\newblock \href {https://doi.org/10.1007/s00023-020-00968-4}
  {\path{doi:10.1007/s00023-020-00968-4}}.

\bibitem{Nagaoka1965}
Y.~Nagaoka.
\newblock {Ground state of correlated electrons in a narrow almost half-filled
  $s$ band}.
\newblock {\em Solid State Communications}, 3(12):409--412, Dec. 1965.
\newblock \href {https://doi.org/10.1016/0038-1098(65)90266-8}
  {\path{doi:10.1016/0038-1098(65)90266-8}}.

\bibitem{Norris1997}
J.~R. Norris.
\newblock {\em Markov Chains}.
\newblock Cambridge University Press, Feb. 1997.
\newblock \href {https://doi.org/10.1017/cbo9780511810633}
  {\path{doi:10.1017/cbo9780511810633}}.

\bibitem{Pagano2014}
G.~Pagano, M.~Mancini, G.~Cappellini, P.~Lombardi, F.~Sch\"{a}fer, H.~Hu, X.-J.
  Liu, J.~Catani, C.~Sias, M.~Inguscio, and L.~Fallani.
\newblock {A one-dimensional liquid of fermions with tunable spin}.
\newblock {\em Nature Physics}, 10(3):198--201, Feb. 2014.
\newblock \href {https://doi.org/10.1038/nphys2878}
  {\path{doi:10.1038/nphys2878}}.

\bibitem{PhysRevB.96.075149}
L.~Pan, Y.~Liu, H.~Hu, Y.~Zhang, and S.~Chen.
\newblock {Exact ordering of energy levels for one-dimensional interacting
  Fermi gases with $SU$($N$) symmetry}.
\newblock {\em Phys. Rev. B}, 96:075149, Aug 2017.
\newblock \href
  {https://doi.org/10.1103/PhysRevB.96.075149}
  {\path{doi:10.1103/PhysRevB.96.075149}}.

\bibitem{PhysRevB.77.144520}
A.~Rapp, W.~Hofstetter, and G.~Zar\'and.
\newblock {Trionic phase of ultracold fermions in an optical lattice: A
  variational study}.
\newblock {\em Phys. Rev. B}, 77:144520, Apr 2008.
\newblock \href
  {https://doi.org/10.1103/PhysRevB.77.144520}
  {\path{doi:10.1103/PhysRevB.77.144520}}.

\bibitem{PhysRevLett.98.160405}
A.~Rapp, G.~Zar\'and, C.~Honerkamp, and W.~Hofstetter.
\newblock {Color Superfluidity and ``Baryon'' Formation in Ultracold Fermions}.
\newblock {\em Phys. Rev. Lett.}, 98:160405, Apr 2007.
\newblock 
  \href {https://doi.org/10.1103/PhysRevLett.98.160405}
  {\path{doi:10.1103/PhysRevLett.98.160405}}.

\bibitem{Scazza2014}
F.~Scazza, C.~Hofrichter, M.~H\"{o}fer, P.~C.~D. Groot, I.~Bloch, and
  S.~F\"{o}lling.
\newblock {Observation of two-orbital spin-exchange interactions with ultracold
  {SU}($N$)-symmetric fermions}.
\newblock {\em Nature Physics}, 10(10):779--784, Aug. 2014.
\newblock \href {https://doi.org/10.1038/nphys3061}
  {\path{doi:10.1038/nphys3061}}.

\bibitem{Tamura2021}
K.~Tamura and H.~Katsura.
\newblock {Ferromagnetism in d-Dimensional {SU}($n$) Hubbard Models with Nearly
  Flat Bands}.
\newblock {\em Journal of Statistical Physics}, 182(1), Jan. 2021.
\newblock \href {https://doi.org/10.1007/s10955-020-02687-w}
  {\path{doi:10.1007/s10955-020-02687-w}}.

\bibitem{Tasaki2020}
H.~Tasaki.
\newblock {\em {Physics and Mathematics of Quantum Many-Body Systems}}.
\newblock Springer International Publishing, 2020.
\newblock \href {https://doi.org/10.1007/978-3-030-41265-4}
  {\path{doi:10.1007/978-3-030-41265-4}}.

\bibitem{Thouless_1965}
D.~J. Thouless.
\newblock {Exchange in solid 3He and the Heisenberg Hamiltonian}.
\newblock {\em Proceedings of the Physical Society}, 86(5):893--904, nov 1965.
\newblock \href {https://doi.org/10.1088/0370-1328/86/5/301}
  {\path{doi:10.1088/0370-1328/86/5/301}}.

\bibitem{Titvinidze2011}
I.~Titvinidze, A.~Privitera, S.-Y. Chang, S.~Diehl, M.~A. Baranov, A.~Daley,
  and W.~Hofstetter.
\newblock {Magnetism and domain formation in {SU}($3$)-symmetric multi-species
  Fermi mixtures}.
\newblock {\em New Journal of Physics}, 13(3):035013, Mar. 2011.
\newblock \href {https://doi.org/10.1088/1367-2630/13/3/035013}
  {\path{doi:10.1088/1367-2630/13/3/035013}}.

\bibitem{Yoshida2021}
H.~Yoshida and H.~Katsura.
\newblock {Rigorous Results on the Ground State of the Attractive {SU}($N$)
  Hubbard Model}.
\newblock {\em Physical Review Letters}, 126(10), Mar. 2021.
\newblock \href {https://doi.org/10.1103/physrevlett.126.100201}
  {\path{doi:10.1103/physrevlett.126.100201}}.

\bibitem{Zhang2014}
X.~Zhang, M.~Bishof, S.~L. Bromley, C.~V. Kraus, M.~S. Safronova, P.~Zoller,
  A.~M. Rey, and J.~Ye.
\newblock {Spectroscopic observation of {SU}($N$)-symmetric interactions in Sr
  orbital magnetism}.
\newblock {\em Science}, 345(6203):1467--1473, Sept. 2014.
\newblock \href {https://doi.org/10.1126/science.1254978}
  {\path{doi:10.1126/science.1254978}}.

\bibitem{Zhao2007}
J.~Zhao, K.~Ueda, and X.~Wang.
\newblock {Insulating Charge Density Wave for a Half-Filled {SU}($N$) Hubbard
  Model with an Attractive On-Site Interaction in One Dimension}.
\newblock {\em Journal of the Physical Society of Japan}, 76(11):114711, Nov.
  2007.
\newblock \href {https://doi.org/10.1143/jpsj.76.114711}
  {\path{doi:10.1143/jpsj.76.114711}}.

\end{thebibliography}
\end{document}